\documentclass[11pt, onecolumn]{IEEEtran}

\usepackage{amsmath}
\usepackage{amssymb}
\usepackage{amsfonts}
\usepackage{epsfig}
\usepackage{subfigure}
\usepackage{calc}
\usepackage{color}
\usepackage[all]{xy}
\usepackage{bm}
\usepackage{algorithmicx}
\usepackage{algpseudocode}
\usepackage{algorithm}
\usepackage{enumerate}
\usepackage{hyperref}
\usepackage{pdfsync}

\newtheorem{defn}{Definition}

\newtheorem{thm}{Theorem}[section]
\newtheorem{cor}[thm]{Corollary}
\newtheorem{prop}{Proposition}

\newtheorem{lem}[thm]{Lemma}
\newtheorem{conj}[thm]{Conjecture}
\newtheorem{constr}[thm]{Construction}
\newtheorem{note}{Remark}
\newtheorem{example}{Example}
\newcommand{\bit}{\begin{itemize}}
\newcommand{\eit}{\end{itemize}}
\newcommand{\bcor}{\begin{cor}}
\newcommand{\ecor}{\end{cor}}
\newcommand{\beq}{\begin{equation}}
\newcommand{\eeq}{\end{equation}}
\newcommand{\beqn}{\begin{equation*}}
\newcommand{\eeqn}{\end{equation*}}
\newcommand{\bea}{\begin{eqnarray}}
\newcommand{\eea}{\end{eqnarray}}
\newcommand{\bean}{\begin{eqnarray*}}
\newcommand{\eean}{\end{eqnarray*}}
\newcommand{\ben}{\begin{enumerate}}
\newcommand{\een}{\end{enumerate}}
\newcommand{\bdefn}{\begin{defn}}
\newcommand{\edefn}{\end{defn}}
\newcommand{\bnote}{\begin{note}}
\newcommand{\enote}{\end{note}}
\newcommand{\bprop}{\begin{prop}}
\newcommand{\eprop}{\end{prop}}
\newcommand{\blem}{\begin{lem}}
\newcommand{\elem}{\end{lem}}
\newcommand{\bthm}{\begin{thm}}
\newcommand{\ethm}{\end{thm}}
\newcommand{\bconj}{\begin{conj}}
\newcommand{\econj}{\end{conj}}
\newcommand{\bconstr}{\begin{constr}}
\newcommand{\econstr}{\end{constr}}
\newcommand{\bpf}{\begin{proof}}
\newcommand{\epf}{\end{proof}}

\title{Codes with Locality for Two Erasures}
\author{N. Prakash, V. Lalitha and P. Vijay Kumar
\thanks{N. Prakash, V. Lalitha and P. Vijay Kumar are with the Department of ECE, Indian Institute of Science, Bangalore,
560 012 India (email: \{prakashn, lalitha, vijay\}@ece.iisc.ernet.in).}
\thanks{This research is supported in part by the National Science Foundation under Grant 0964507 and in part by the NetApp
Faculty Fellowship program. The work of V. Lalitha is supported by a TCS Research Scholarship. } }
\begin{document}

%\tableofcontents

\maketitle

\begin{abstract}
 In this paper, we study codes with locality that can recover from two erasures via a sequence of two local, parity-check computations. By a local parity-check computation, we mean recovery via a single parity-check equation associated to small Hamming weight.  Earlier approaches considered recovery in parallel; the sequential approach allows us to potentially construct codes with improved minimum distance.  These codes, which we refer to as locally 2-reconstructible codes, are a natural generalization along one direction, of codes with all-symbol locality introduced by Gopalan \textit{et al}, in which recovery from a single erasure is considered. By studying the Generalized Hamming Weights of the dual code, we derive upper bounds on the minimum distance of locally 2-reconstructible codes and provide constructions for a family of codes based on Tur\'an graphs, that are optimal with respect to this bound. The minimum distance bound derived here is universal in the sense that no code which permits all-symbol local recovery from $2$ erasures can have larger minimum distance regardless of approach adopted. Our approach also leads to a new bound on the minimum distance of codes with all-symbol locality for the single-erasure case.
\end{abstract}

\section{Introduction}\label{sec:intro}

A primary goal in distributed data storage is the efficient repair of a failed node. While regenerating codes~\cite{DimGodWuWaiRam} aim to minimize the amount of data download needed to carry out node repair, codes with locality~\cite{GopHuaSimYek} seek to minimize the number of nodes accessed during node repair.  The focus of the present paper is on codes with locality. 

Let $\mathcal{C}$ denote an $[n, k, d_{\min}]$ linear code having block length $n$, dimension $k$ and minimum distance $d_{\min}$. Where the minimum distance is not relevant, we will simply refer to $\mathcal{C}$ as an $[n,k]$ code.  
The $i^{\text{th}}$ code-symbol $c_i$, $1 \leq i \leq n$, of the code $\mathcal{C}$ is said to have locality $r$ if this symbol can be recovered by accessing at most $r$ other code symbols and performing a linear computation. Equivalently, there exists a row in the parity-check matrix $H$ of the code of Hamming weight $\leq (r + 1)$, whose support includes $i$. A systematic code in which all the $k$ message symbols have locality $r$ is said to have  information locality $r$.  The minimum distance $d_{\min}$ of a code with information locality $r$ is upper bounded~\cite{GopHuaSimYek} by
\begin{equation} \label{eq:gopalan_bound}
d_{\min} \leq   n- k -\left\lceil \frac{k}{r} \right\rceil + 2.
\end{equation}
The pyramid-code construction in~\cite{HuaCheLi} yields optimal codes for all $\{n,k,r\}$ with field size $O(n)$.   
The authors of \cite{GopHuaSimYek} also introduce the notion of all-symbol locality in which all code symbols, not just the message symbols, have locality $r$.  They show the existence of codes with all-symbol locality that achieve the bound in \eqref{eq:gopalan_bound} when $(r+1) \mid n$, but leave open the question as to whether it is possible to derive a tighter bound in the all-symbol-locality case, for general $\{n,k,r\}$.     The all-symbol-locality property is preferable in applications as it permits a uniform approach to storage-system design.  Codes with locality also go by the name locally-repairable~\cite{PapDim}, locally-reconstructible~\cite{HuaSimXu_etal_azure} and locally-recoverable codes~\cite{TamBar}.

\subsection{Handling Multiple Erasures} \label{sec:intro_multiple}

There is current practical interest in the handling of multiple erasures as simultaneous node failures are not uncommon, given the increasing trend towards replacing expensive servers with low-cost commodity servers, the presence of ``hot" nodes etc.  Several approaches to the multiple-erasure case in the context of codes with locality can be found in the literature. 

The authors of \cite{PraKamLalKum} handle multiple erasures by protecting each message symbol with a local code of length $\leq r+\delta-1$ and minimum distance $\geq \delta$, and derive the upper bound 
\begin{equation} \label{eq:rdelta_bound}
d_{\min} \leq n-k+1 - \left (\left\lceil \frac{k}{r} \right\rceil -1\right )(\delta-1).
\end{equation}
Pyramid codes, are once again shown to be optimal. This notion of locality is extended to the case when all code symbols are so protected and the existence of optimal codes with all-symbol locality is shown for the case when $(r+\delta-1)|n$.
Codes having the capability of locally recovering a failed node in the presence of any $\delta -1$ other node failures are also considered in \cite{JuaHolOgg}.  Constructions based on partial geometry are provided and their rates computed. 

A third approach to handling multiple erasures is presented in \cite{WanZha} in which the authors seek to protect each of the $k$ message symbols by $\delta-1$ support-disjoint local parities, each of length $\leq r+1$. The following upper bound is derived:
\begin{equation} \label{eq:chinese_bound}
d_{\min} \leq n-k+1 - \left ( \left \lceil \frac{(k-1)(\delta-1)+1}{(r-1)(\delta-1)+1} \right \rceil - 1\right),
\end{equation}
and the existence of optimal codes is established for the case when $n \geq k(r(\delta-1)+1)$.
The setting is extended to codes with all-symbol locality for handling $2$ erasures, as well.  A square-code construction that achieves the bound in \eqref{eq:chinese_bound} for restricted values of the code dimension $k$ is presented.  A related setting appears in \cite{TamBar}, where once again $\delta-1$ support-disjoint local parities are used for the protection of all code symbols.  Here however, the local parities are permitted to have different lengths. A key feature of this work is that the authors provide constructions of codes in which the code alphabet is small, on the order of the code length.  Lower bounds to the minimum distance of the codes constructed are also provided.
 
A common underlying theme of the prior approaches in \cite{PraKamLalKum}, \cite{JuaHolOgg}, \cite{WanZha}, \cite{TamBar} is that they implicitly assume the need for the recovery of multiple erased symbols in parallel. However, the need for locality does not preclude a sequential approach such as is adopted here. The sequential approach places a less-stringent requirement on the code and potentially allows us to construct codes with improved minimum distance while still enabling local recovery from erasures. In addition, the minimum distance bound derived here is universal in the sense that no code which permits all-symbol local recovery from $2$ erasures can have larger distance regardless of approach adopted.  The exact formulation and our approach to solving the problem are presented in Section~\ref{sec:approach}.

\subsection{Other Related Work}

Explicit constructions of optimal codes with all-symbol locality for the single erasure case are provided in \cite{SilRawVis}, \cite{TamPapDim}, respectively based on Gabidullin maximum rank-distance and Reed-Solomon codes. Families of codes with all-symbol locality with small alphabet size (low field size) are constructed in \cite{TamBar}.  Locality in the context of non-linear codes  is considered in \cite{PapDim}. Codes with local regeneration are considered in \cite{KamPraLalKum}, \cite{SilRawKoyVis}, \cite{KamSil_etal}. Studies on the implementation and performance evaluation of codes with locality can be found in \cite{HuaSimXu_etal_azure, sathiamoorthy}.

Section~\ref{sec:ghw} provides background on generalized Hamming weights (GHW).  Our formulation and approach to the problem are outlined in Section~\ref{sec:approach}. An important connection between the $k$-cores of \cite{GopHuaSimYek} and GHW is made in Section~\ref{sec:kcore_ghw}. The upper bound on $d_{\min}$ and optimal code constructions can be found in Sections~\ref{sec:bounds} and \ref{sec:optimal_codes} respectively.  The final section, Section~\ref{sec:all_symbol_single} presents the analogous $d_{\min}$ bound for the single-erasure case.   The proofs of most statements appear in the Appendix.

%Before discussing the problem formulation and our approach to solving the problem in Section \ref{sec:approach}, in the next section, we briefly review the notions of Generalized Hamming Weights and gap numbers of a code. These concepts will be central to the rest of the discussion.

\section{Generalized Hamming Weights} \label{sec:ghw}

\begin{defn}
The $i^{th}$, $1 \leq i \leq k$, GHW~\cite{Wei,HelKloLevYtr} of an $[n, k]$ code ${\cal C}$ is the cardinality of the minimum support of an $i$-dimensional subcode of ${\cal C}$, i.e., 
\begin{equation}
d_i({\cal C}) \ = \ d_i \ = \  \min_{\substack{ \mathcal{D} < \mathcal{C} \\ \text{dim}(\mathcal{D}) = i }}
\left|\text{supp}({\cal D}) \right| ,
\end{equation}
where $\mathcal{D} < \mathcal{C}$, is used to denote a subcode $\mathcal{D}$ of $\mathcal{C}$ and $\text{supp}({\mathcal{D}}) \triangleq \cup_{\bf{c} \in \mathcal{D}} \text{supp}(\bf{c})$. 
%$\text{supp}({\cal D})$ denotes the support of the linear subcode $\mathcal{D}$ of $\mathcal{C}$, given by $\text{supp}({\mathcal{D}}) \triangleq \cup_{\bf{c} \in \mathcal{D}} \text{supp}(\bf{c})$. 
\end{defn}

\vspace{0.1in}

It is well known that 
\begin{eqnarray}
d_{\min}(\mathcal{C}) = d_1 < d_2 < \ldots < d_k = n.\label{eq:ordering_GHW}
\end{eqnarray}   
The complement of the set $\{d_i, 1 \leq i \leq k\}$, in $[n]$, will be termed as the set of \textit{gap numbers} (more simply,
gaps) of the code ${\cal C}$ and denoted by $\{ g_i, \ 1 \leq i \leq n-k \}$, i.e.,
\begin{equation} 
\{ g_i, \ 1 \leq i \leq n-k \} \ = \ [n] \setminus \{ d_i, \ 1 \leq i \leq k \}.
\end{equation}
Similarly, let $\{d_j^{\perp}, \ 1 \leq j \leq n-k \}$ and $\{ g_i^{\perp}, \ 1 \leq i \leq k \}$ denote the GHWs and gaps of the dual code ${\cal C}^\perp$. The lemma below~\cite{Wei} relates the GHWs of $\mathcal{C}$ to the gaps of $\mathcal{C}^{\perp}$.

\vspace{0.1in}

\begin{lem} \label{lem:GHW_code_dual_relation}
\bean \label{eq:GHW_code_dual_relation}
d_i & = &  (n+1)-g^{\perp}_{k-i+1} , \ \ 1 \leq i \leq k, \\
d_{\min}(\mathcal{C}) & = & d_1 \ = \ (n+1)- g^{\perp}_{k}.
\eean 
\end{lem}

\vspace{0.1in}
We also note that if $\mathcal{B}_{0} < \mathcal{C}^{\perp}$, then $d_{i}(\mathcal{B}_{0}) \geq d_{i}(\mathcal{C}^{\perp})$, which implies that $g_k^{\perp} \geq g_k(\mathcal{B}_{0})$.  We thus obtain: 
\begin{eqnarray} \label{eq:GHW_subcode_upperbound}
 d_{\min}(\mathcal{C}) & \leq & n + 1 - g_k(\mathcal{B}_{0}).
\end{eqnarray}

\section{Approach and Results} \label{sec:approach} 

Our focus in this paper, is on codes with all-symbol locality for the two-erasure case. 
\vspace{0.1in}
\begin{defn} \label{defn:locality}
A code $\mathcal{C}$ will be said to be \textit{locally $2$-reconstructible with locality $r$},  if for any pair of code-symbol erasures, there exists a sequence of two local (and linear) parity-check computations that can be used to recover the erased symbols. By a local parity-check computation, we mean recovery via a parity whose support covers the coordinate being recovered and involves at most $r$ other code symbols. 
\end{defn}  
\vspace{0.1in}
Note that under the above definition, it is permissible that the symbol recovered by the first local parity belong to the set of $r$ symbols accessed by the second local parity. The families of all-symbol locality codes constructed in  ~\cite{PraKamLalKum}, \cite{JuaHolOgg}, \cite{WanZha} for the case $\delta=3$ may all be regarded as examples of locally $2$-reconstructible codes.  In the sequel, we will refer to a locally $2$-reconstructible code with locality $r$ simply as a locally reconstructible code.  Our principal results are an upper bound on the minimum distance of locally reconstructible codes and optimal constructions for a large class of code parameters. The steps involved in the derivation of the upper bound on $d_{\min}$ are outlined below.

Given a locally reconstructible code $\mathcal{C}$, let $\mathcal{B}_{0}$ denote the subcode 
of the dual code $\mathcal{C}^{\perp}$, spanned by all codewords ${\bf c} \in \mathcal{C}^{\perp}$ of Hamming weight less than or equal to $r+1$, i.e.,
\begin{eqnarray} \label{eq:C0_local}
\mathcal{B}_{0} & = & \text{span}\left({\bf c} \in \mathcal{C}^{\perp}, |\text{supp}({\bf c})| \leq r+1 \right).
\end{eqnarray}

\noindent \textit{a) Step $1$}: We first establish that the dimension $b$ of $\mathcal{B}_{0}$ satisfies the lower bound $b \  \geq \ \frac{2n}{r+2}$.

\vspace{0.1in}

\noindent \textit{b) Step $2$}: Next, we observe from \eqref{eq:GHW_subcode_upperbound} that the minimum distance of the code ${\cal C}$ satisfies $d_{\min}({\cal C}) \ \leq \ n+1-g_k(\mathcal{B}_0)$.

%It is known from the theory of Generalized Hamming weights (GHWs) that the minimum distance $d_{\min}$ of the code ${\cal C}$ is given by $d_{\min}({\cal C}) \ = \ n+1-g_k(\mathcal{C}^{\perp})$, where $g_k(\mathcal{C}^{\perp})$ denotes the $k^{\text{th}}$ gap number of  the code $\mathcal{C}^{\perp}$. We refer to Appendix \ref{sec:ghw} for a quick overview of the notions of GHWs that we use in this paper. The theory of GHW also tells us whenever $\mathcal{B}_0$ is a subcode of ${\cal C}^{\perp}$, we have $g_k(\mathcal{B}_0) \ \leq \ g_k(\mathcal{C}^{\perp})$, leading to the bound

\vspace{0.1in}

\noindent \textit{c) Step $3$}: We then obtain a lower bound on the $k^{\text{th}}$ gap of $\mathcal{B}_0$ of the form $g_k(\mathcal{B}_0) \geq \gamma_k$, leading to the desired upper bound $d_{\min}(\mathcal{C}) \ \leq \ n+1- \gamma_k$. 
 This step makes use of the lower bound on the dimension $b$ of $\mathcal{B}_0$, derived in Step 1. 
%The upper bounds on the GHWs of $\mathcal{B}_0$ is  established via a lemma regarding set unions, which states that given a set which is a union of $b$ subsets, each of cardinality $r+1$, one can obtain upper bounds on the minimum cardinality of the union of any $m$ subsets (for some $m$), in an iterative manner. A consequence of this method is that the bounds on the GHWs of $\mathcal{B}_0$ will depend only the code length $n$ and the locality parameter $r$. The dimension $k$ of the code $\mathcal{C}$ plays a role only when we count the $k^{\text{th}}$ gap given the $b$ GHWs. 

\vspace{0.1in}

The same sequence of steps is also applied in Section~\ref{sec:all_symbol_single} to the case of codes with all-symbol locality for the single-erasure case. This results in a new bound on $d_{\min}$ for this class of codes, tighter in general than that given by \eqref{eq:gopalan_bound}. 

\vspace{-0.15in}

\subsection{Optimal Constructions} We provide code constructions that are optimal with respect to the bound on $d_{\min}$ given above in Step 3 whenever the block length $n$ is of the form $n = \frac{(r+\beta)(r+2)}{2}, 1 \leq \beta \leq r$, with $\beta|r$.  The steps involved are described below.

\vspace{0.1in}

\noindent \textit{a) Step $1$}: We begin by constructing a code $\mathcal{B}_{0}$ such that $(i)$ the code formed by the null space of ${\cal B}_0$ possesses the locally reconstructible property, $(ii)$ $\dim (\mathcal{B}_{0}) = b = \frac{2n}{r+2}$, and $(iii)$ the lower bound on the $k^{\text{th}}$ gap is also achieved, i.e., $g_k(\mathcal{B}_0) = \gamma_k.$ Our construction of $\mathcal{B}_{0}$ is based on Tur\'an graphs~\cite{Die},  depends only on code length $n$ and locality parameter $r$, and is independent of the dimension $k$ of the desired code ${\cal C}$. 

\vspace{0.1in}

\noindent \textit{b) Step $2$}: Given the code $\mathcal{B}_{0}$, it turns out that it is always possible to find 
an $[n, k]$ code $\mathcal{C}$ such that ${\cal B}_0$ is a subcode of ${\cal C}^{\perp}$, $g_k(\mathcal{C}^{\perp}) = g_k(\mathcal{B}_0)$ and this code $\mathcal{C}$ is then the desired locally reconstructible code.  It has the best possible minimum distance given by $d_{\min}(\mathcal{C}) = n + 1 - \gamma_k$. 

This proof is an instance of a more general result that is important in its own right and which can potentially be applied in other situations as well.  It combines the notion of a $k$-core introduced in \cite{GopHuaSimYek} with the GHW structure of a code to enable construction of the best possible code of a given dimension when the code is linearly constrained.  This is discussed in more detail in the next section. 
%relies on a more general result which states that given any code $\mathcal{B}_{0}$, one can always identify a super code $\mathcal{C}^{\perp}$ of which  $\mathcal{B}_{0}$ is a subcode such that the minimum distance of the dual code $\mathcal{C}$ of $\mathcal{C}^{\perp}$, is as large as possible, i.e, $d_{\min}(\mathcal{C}) = n + 1 - g_k(\mathcal{B}_0)$.  The notion of a $k$-core of a code as defined in \cite{GopHuaSimYek} is used to establish this fact.  
%%
%The organization of the rest of the document is as follows. In section \ref{sec:kcore_ghw}, we review the definition of $k$-cores of a code and establish a connection between GHWs and $k$-cores. The description of the minimum distance bound and the optimal code construction respectively appear in Sections \ref{sec:bounds} and \ref{sec:optimal_codes}. The new $d_{\min}$ bound for codes with all-symbol locality, in the single erasure case is finally presented in Section \ref{sec:all_symbol_single}. 

\section{$k$-cores and Connection with GHWs}\label{sec:kcore_ghw}

\begin{defn}[\cite{GopHuaSimYek}] \label{def:k-core}
Given a linear code $\mathcal{B}_{0}$, a set $S \subseteq [n], |S| = \ell$ is termed an $\ell$-core of $\mathcal{B}_{0}$ if
$\text{supp}({\bf c}) \nsubseteq S $, $\forall {\bf c} \in \mathcal{B}_{0}$.
\end{defn}

\vspace{0.1in}

The above definition is equivalent to saying that $\text{rank}\left(G'|_S\right) = \ell$, for any $S$ which is an $\ell$-core of $\mathcal{B}_{0}$, where $G'$ denotes a generator matrix of $\mathcal{B}_{0}^{\perp}$. The lemma below was used in \cite{GopHuaSimYek} to show the existence of all-symbol locality codes when $(r+1)|n$, and will also prove very useful here. 

\vspace{0.1in}

\begin{lem}[\cite{GopHuaSimYek}]\label{lem:k-core_existence}
Let $\mathcal{B}_{0}$ denote an $[n, t]$ code over $\mathbb{F}_q$. Then for any $k$ such that $k \leq n - t$, there
exists an $[n, k]$ code $\mathcal{C}$ over $\mathbb{F}_q$ such that
\begin{enumerate}[(a)]
 \item $\mathcal{B}_{0} < \mathcal{C}^{\perp}$, and
 \item any $S$ which is a $k$-core of $\mathcal{B}_{0}$ is also a $k$-core of $\mathcal{C}^{\perp}$,
\end{enumerate}
whenever $q > kn^k$.
\end{lem}

\vspace{0.1in}

In the following theorem, we obtain an expression for the minimum distance of the code whose existence is guaranteed
by the above lemma.

\vspace{0.1in}

\begin{thm} \label{thm:GHW_k-core_dmin}
Let $\mathcal{B}_{0}$ denote an $[n, t]$ code and let $\mathcal{C}$ denote an $[n, k]$ code, $k \leq n - t$, such that
\begin{enumerate}[(a)]
 \item $\mathcal{B}_{0} < \mathcal{C}^{\perp}$, and
 \item any $S$ which is a $k$-core of $\mathcal{B}_{0}$ is also a $k$-core of $\mathcal{C}^{\perp}$.
\end{enumerate}
The minimum distance of the code $\mathcal{C}$ is given by 
\begin{eqnarray}
 d_{\min}(\mathcal{C}) & = & n + 1 - g_{k}(\mathcal{B}_{0}).
\end{eqnarray}
\end{thm}
\begin{proof}
 See Appendix \ref{app:GHW_k-core_dmin}.
\end{proof}

\vspace{0.1in}

Note from \eqref{eq:GHW_subcode_upperbound} that the minimum distance of the code $\mathcal{C}$, whenever $\mathcal{B}_{0} <
\mathcal{C}^{\perp}$, cannot be any larger than $n + 1 - g_{k}(\mathcal{B}_{0})$. In addition to showing that the $k^{\text{th}}$ gap of $\mathcal{C}^{\perp}$ is same as that of $\mathcal{B}_{0}$, it is possible to identify all the GHWs of $\mathcal{C}^{\perp}$ in terms of the GHWs of $\mathcal{B}_{0}$. This is stated in the following theorem.

\vspace{0.1in}

\begin{thm} \label{thm:GHW_k-core}
Let $\mathcal{B}_{0}$ denote an $[n, t]$ code and let $\mathcal{C}$ denote an $[n, k]$ code, $k \leq n - t$, such that
\begin{enumerate}[(a)]
 \item $\mathcal{B}_{0} < \mathcal{C}^{\perp}$, and
 \item any $S$ which is a $k$-core of $\mathcal{B}_{0}$ is also a $k$-core of $\mathcal{C}^{\perp}$.
\end{enumerate}
The generalized Hamming weights of $\mathcal{C}^{\perp}$ are given by
\begin{eqnarray} \label{eq:GHW_completion_isit}
d_i(\mathcal{C}^{\perp}) & = & \left\{ \begin{array}{c} d_i(\mathcal{B}_{0}), \ 1 \leq i \leq g_k(\mathcal{B}_{0}) - k \\
i + k, \ g_k(\mathcal{B}_{0}) - k +1 \leq i \leq n - k \end{array}. \right.
\end{eqnarray}
\end{thm}
\begin{proof}
 See Appendix \ref{app:GHW_k-core}.
\end{proof}

\vspace{0.1in}

An illustration of \eqref{eq:GHW_completion_isit} is given in Fig. \ref{fig:kcore_poles} with parameters $n=15$, $t=5$ and $k=8$. We see that the largest gap of  $\mathcal{C}^{\perp}$ is same as the $k^{\text{th}}$ gap of $\mathcal{B}_{0}$. Moreover the GHWs of $\mathcal{C}^{\perp}$ which appear to the left of the $k^{\text{th}}$  gap are exactly same as those of $\mathcal{B}_{0}$.

\begin{center}
\begin{figure} [h!]
\begin{center}
\includegraphics[width=4in]{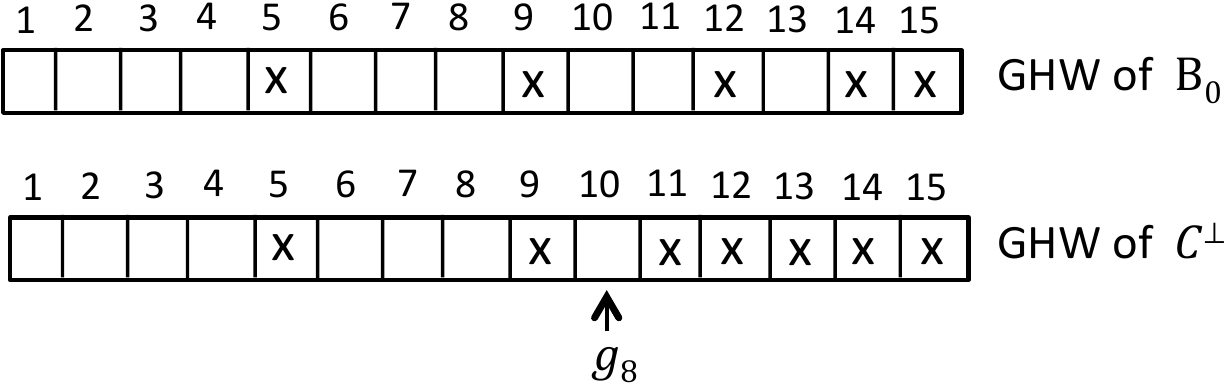}
\end{center}
\caption{Relation between GHWs of $\mathcal{C}^{\perp}$ and  $\mathcal{B}_{0}$, with $n=15$, $t=5$ and $k=8$. GHWs are indicated by an `X', gaps by a blank.}
\label{fig:kcore_poles}
\end{figure}
\end{center}

\section{Minimum Distance Bound for Locally Reconstructible Codes} \label{sec:bounds}

In this section, we will obtain upper bounds on the GHWs of the subcode $\mathcal{B}_{0}$, as defined in \eqref{eq:C0_local}. This in turn will establish a lower bound on the $k^{\text{th}}$ gap $g_k(\mathcal{B}_{0})$ from which we will obtain an upper bound on the minimum distance of $\mathcal{C}$. We begin with a characterization of a locally reconstructible code.

\vspace{0.1in}

\begin{lem} \label{lem:loc_char_A}
Let $\mathcal{A}_i$ denote the collection of all the local parities which cover the code symbol $c_i, 1 \leq i \leq n$. Code $\mathcal{C}$ is locally reconstructible iff $(i)$ $|\mathcal{A}_i|  \ \geq \ 1$ and $(ii)$ $\mathcal{A}_i \ \neq \  \mathcal{A}_j, \ \forall i, j, \ i\neq j$.
\end{lem}
\begin{proof}
 Straightforward.
\end{proof}

\vspace{0.1in}

The parallel recovery of two code symbols $c_i$ and $c_j$ is possible iff $\mathcal{A}_i \nsubseteq \mathcal{A}_j$ and $\mathcal{A}_j \nsubseteq \mathcal{A}_i$.  In the event that $\mathcal{A}_i \subsetneq \mathcal{A}_j$, $c_j$ can be recovered first through a local computation not involving $c_i$ and having recovered $c_j$, $c_i$ can then be recovered.  

\vspace{0.1in}

\begin{thm} \label{thm:dim_local_code}
The dimension of the subcode $\mathcal{B}_{0}$ defined in \eqref{eq:C0_local} is lower bounded by
\begin{eqnarray}
\text{dim}(\mathcal{B}_{0}) &\geq &\frac{2n}{r+2}. 
\end{eqnarray}
\end{thm}
\begin{proof}
 See Appendix \ref{app:dim_local_code}.
\end{proof}

\vspace{0.1in}

\begin{cor}\label{cor:rate_bound}
The rate of any code $\mathcal{C}$ which is locally reconstructible is upper bounded by
\begin{eqnarray} \label{eq:rate_bound}
 \frac{k}{n} & \leq & \frac{r}{r+2}.
\end{eqnarray}
\end{cor}

\vspace{0.1in}

The following lemma will be used to establish upper bounds on the GHWs of $\mathcal{B}_{0}$.

\vspace{0.1in}

\begin{lem}\label{lem:min_set_union}
Let $T$ be any set such that $|T| = n \geq r+1$ and let $S_i, \ 1 \leq i \leq b$ be subsets of $T$
such that $(i)$ $\cup_{i=1}^{b}S_i = T$ and $(ii)$ $|S_i| = r+1, \forall i \in \{1, 2, \ldots, b\}$.
Define

\begin{eqnarray} \label{eq:fm}
f_{m} = \min_{\substack{I \subseteq [b] \\ |I|= m}} \left| \cup_{i \in I}S_i \right|, 1 \leq m \leq b.
\end{eqnarray}

Then, $\forall m \in [b]$,  $f_{m} \leq e_m$, where the $\{e_m\}$ are obtained recursively as follows:
\vspace{-0.1in}
\begin{eqnarray}
e_b & = & n, \label{eq:e_def1}\\
e_{m-1} & = & e_{m} - \left\lceil \frac{2e_m}{m}\right\rceil + (r+1),  \ 2 \leq m \leq b. \label{eq:e_def2}
\end{eqnarray}
\end{lem}
\begin{proof}
See Appendix \ref{app:min_set_union}.
\end{proof}

\vspace{0.1in}

Note that in Lemma \ref{lem:min_set_union}, since $\cup_{i=1}^{b}S_i = T$, we have that $b \geq \frac{n}{r+1}$ and thus setting $m=b$ in \eqref{eq:e_def2} and dropping the ceiling function, we obtain 
\begin{eqnarray}
e_{b-1} & \leq & (1 - \frac{2}{b})e_b + (r+1) \ = \ (1 - \frac{2}{b})n + (r+1), \nonumber\\
& \hspace{-0.5in} \leq & \hspace{-0.3in} (1 - \frac{2}{b})b(r+1) + (r+1) \ = \ (b-1)(r+1). \label{eq:comparision_bounds_single}
\end{eqnarray}
The arguments in \eqref{eq:comparision_bounds_single} can be iterated further (with $m=b-1$ and so on) and from this it follows that Lemma \ref{lem:min_set_union} indeed implies the obvious bound $f_m \leq m(r+1), 1 \leq m \leq b$. In general the bounds given by Lemma \ref{lem:min_set_union} will be tighter than this and will take into account the fact that the total support has cardinality only $n$.

\vspace{0.1in}

\begin{thm} \label{thm:GHW_upper_bounds_C0}
Let  $\mathcal{C}$ denote an $[n, k]$ locally reconstructible code and let $\mathcal{B}_{0}$ be the subcode as defined in \eqref{eq:C0_local}. Set $b = \left\lceil \frac{2n}{r+2}\right\rceil$. Then the first $b$ GHWs of $\mathcal{B}_{0}$, and hence those of $\mathcal{C}^{\perp}$, are upper bounded by $d_m(\mathcal{B}_{0}) \leq e_m, \ 1 \leq m \leq b$, where $e_m$ is as defined by \eqref{eq:e_def1}
and \eqref{eq:e_def2}. Furthermore, if $\ell$ denotes the unique integer satisfying 
$e_{\ell} < k+\ell < e_{\ell+1}$, then the minimum distance of $\mathcal{C}$ is upper bounded by
\begin{eqnarray}
d_{\min}(\mathcal{C}) & \leq & n + 1 - (k + \ell). \label{eq:d_min_bound_seq}
\end{eqnarray}
\end{thm}
\begin{proof}
Consider a basis of $\mathcal{B}_{0}$ which are composed only of codewords of Hamming weight less than or equal to $r+1$.
Let $\{S_i, i = 1, \ldots, b\}$ denote their supports, where $b \geq \frac{2n}{r+2}$. The bounds on the GHWs of $\mathcal{B}_0$ now follows directly upon applying Lemma \ref{lem:min_set_union} to the sets $\{S_i, i = 1, \ldots, b\}$. (Note that if any set $S_i$ has cardinality less than $r+1$, one can simply substitute $S_i$ with any set $S_i'$ such that $|S_i'| = r+1, S_i \subseteq S_i'$ and then apply Lemma \ref{lem:min_set_union}.). 

Given the bounds on the GHWs of $\mathcal{B}_0$, the $k^{\text{th}}$ gap of $\mathcal{B}_{0}$ is lower bounded by  $g_k(\mathcal{B}_{0}) \geq k + \ell$, where $\ell$ denotes the unique integer such that $e_{\ell} < k+\ell < e_{\ell+1}$ (to see this, assume that first $b$ GHWs are given exactly by the sequence $\{e_m, m = 1, \ldots, b\}$ and using this, identify the $k$th gap). The bound on $d_{\min}$ finally follows from  \eqref{eq:GHW_subcode_upperbound}.
\end{proof}

\vspace{0.1in}

A code $\mathcal{C}$ will be called an \textit{optimal locally reconstructible code} if it achieves the bound in
\eqref{eq:d_min_bound_seq} with equality.

\section{Optimal Locally Reconstructible Codes} \label{sec:optimal_codes}

In this section, we will describe a construction for optimal locally reconstructible codes for the case when the length of the code takes on the form
\begin{eqnarray}
n & = & \frac{(r+\beta)(r+2)}{2},
\end{eqnarray}
with $1 \leq \beta \leq r$ and $\beta|r$. The only restriction on the dimension $k$ is the necessary rate restriction given by Corollary \ref{cor:rate_bound}, i.e., $k \leq \frac{rn}{r+2}$.  As described in Section \ref{sec:intro}, our approach to optimal code construction will involve first constructing a code $\mathcal{B}_{0}$ which depends only on $n, r$ and is independent of $k$. The construction of  $\mathcal{B}_{0}$ will be based on Tur\'an graphs and will be such that 
\begin{enumerate}[(a)]
\item $\mathcal{B}_{0}^{\perp}$ is  locally reconstructible,
\item $\text{dim}(\mathcal{B}_{0}) = b = \frac{2n}{r+2} = r + \beta$, and
\item all the $b$ GHWs of $\mathcal{B}_{0}$ achieve the upper bounds given by Theorem
\ref{thm:GHW_upper_bounds_C0}.
\end{enumerate}
Once we have the code $\mathcal{B}_{0}$, the desired $[n,k]$ code is simply the code $\mathcal{C}$, whose existence is guaranteed by Lemma \ref{lem:k-core_existence}. It is clear, based on the discussion in Section \ref{sec:bounds} and from Theorem \ref{thm:GHW_k-core} that this code $\mathcal{C}$ will be an optimal locally reconstructible code.

\subsection{Construction of $\mathcal{B}_{0}$ Using Tur\'an Graphs}

Consider a graph with $b = \frac{2n}{r+2} = r + \beta$ vertices.  We partition the vertices into $x = \frac{r+\beta}{\beta}$
partitions, each partition containing $\beta$ vertices.  We next place exactly one edge between any two vertices belonging to two
distinct partitions. The resulting graph is known as a Tur\'an graph on $b$ vertices with $x$ vertex partitions. The number of edges in this graph is $\frac{x(x-1)\beta^2}{2} = n - b$ and each vertex is connected to exactly $(x-1)\beta = r$ other vertices. Let the vertices be labelled from $1$ to $b$ and the edges be labelled from $b + 1$ to $n$, without paying attention to order.

To convert the graph into a code, we proceed as follows. Associate a local parity with each of the $b$ vertices, let parity
$\underline{p}_i$ be associated with vertex $i, \ 1 \leq i \leq b$. Let $\{i_1, i_2, \ldots, i_r\}$ denote all the edges
which are incident up on vertex $i$. Then, the support $S_i \subseteq [n]$ of the local parity $\underline{p}_i$ is set at 
\begin{eqnarray} \label{eq:supports_turan}
 S_i & = & \{i, i_1, i_2, \ldots, i_r\}
\end{eqnarray}
and the codeword ${\bf c}_i$ corresponding to $\underline{p}_i$ is identified as the all-$1$ vector in these $r+1$ coordinates
(with zeros in the remaining $n-(r+1)$ coordinates). Set $\mathcal{B}_{0} = \text{span}({\bf c}_i, \ 1 \leq i
\leq b)$. It is easily verified that the code $\mathcal{B}_{0}^{\perp}$ is locally reconstructible and that its dual $\mathcal{B}_{0}$ has dimension $b = \frac{2n}{r+2} = r + \beta$. Before proceeding to evaluate the GHWs of the code $\mathcal{B}_{0}$ and proving their optimality w.r.t. Theorem \ref{thm:GHW_upper_bounds_C0}, we first illustrate the construction using two examples.

\vspace{0.1in}

\begin{example}\label{ex:beta_1}
Consider the parameters $r=3$ and $\beta = 1$, which implies that the length $n = 10$. When $\beta = 1$, note that the
number of partitions $x = r + 1 = 4$, and each partition has just one vertex. Thus the total number of vertices $b = r+1 = 4$
and the graph is simply a completely connected graph on $b=4$ vertices. The generator matrix of the code $\mathcal{B}_{0}$
(upto permutation of columns) in this case is given by
\begin{eqnarray}
 H_0 & = & \left[ \begin{array}{cccccccccc} 1& 1& 1& 1& & & & & & \\
                   1& & & & 1& 1& 1& & & \\
		   & 1& & & 1& & & 1&1 & \\
		   & & 1& & & 1& &1 & &1
                  \end{array} \right]. \nonumber
\end{eqnarray}
Note that the sum of the four rows of $H_0$ gives a local parity having support $\{4, 7, 9, 10\}$. Thus we see that the code $\mathcal{B}_{0}$ guarantees that any code symbol is covered by two support disjoint local parities. In general, when $\beta=1$, the construction gives $(r, \delta = 3)_c$ all-symbol locality codes, with length $n = \frac{(r+1)(r+2)}{2}$ (see Section \ref{sec:intro} for a definition of $(r, \delta)_c$ codes). 
\end{example}

\vspace{0.1in}

\begin{example}\label{ex:beta_r}
In this example, let $\beta = r=3$, which implies that the length $n = 15$. When $\beta = r$, we get a
bipartite graph with $r$ vertices on each of the two partitions. The generator matrix of the code $\mathcal{B}_{0}$, which is the span of $6$ codewords, is given (up to permutation of columns) by

{\footnotesize
\begin{equation}
 H_0  =  \left[ \begin{array}{ccccccccccccccc} 1& 1& 1& 1& & & & & & & & & & &\\
                   & & & & 1& 1& 1& 1& & & & & & & \\
		   & & & & & & & & 1& 1& 1& 1& & & \\
		   1& & & &1 & & & & 1& & & & 1& & \\
		   & 1& & & &1 & & & & 1& & & & 1 &  \\
		   & & 1& & & &1 & & & &1 & & & & 1
                  \end{array} \right]. \nonumber
\end{equation}
}

It can be verified that any non-trivial linear combination (resulting in vectors other than those appearing in the rows of $H_0$) of the $6$ vectors results in a codeword whose Hamming weight $\geq 5$. As a result, each of the code symbols $\{c_4, c_8, c_{12}, c_{13}, c_{14}, c_{15}\}$ is covered by only one local parity and thus parallel decoding of two erased symbols may not always be possible. For instance, if symbols $c_3, c_4$ get erased, we must necessarily decode $c_3$ first before decoding $c_4$.

The Tur\'an graph construction, when $\beta = r$ is closely related to the square code construction presented in
\cite{WanZha}. The square code construction was used to guarantee two support disjoint parities for each code word symbol.
For the current example, the closest relative from the square code family has length $16$ and has one
more local parity (in addition to those described by $\mathcal{B}_{0}$) covering the coordinates $\{4, 8, 12, 16\}$.  A second local parity which covers $c_{16}$ can be obtained as a linear combination of all the $7$ parities, and this will be  a parity on the support $\{13, 14, 15, 16\}$.
\end{example}

\subsection{Generalized Hamming Weights of the Constructed Code $\mathcal{B}_{0}$}

\begin{thm} \label{thm:supports_achievability}
Consider the code $\mathcal{B}_{0}$ obtained via the Tur\'an graph construction along with support sets $\{S_i, \ 1 \leq
i \leq b = r + \beta\}$, as described in \eqref{eq:supports_turan}, associated to the $b$ local parities $\{\underline{p}_i\}$. Then the sets $\{S_i, \ 1 \leq i \leq b = r + \beta\}$ achieve the upper bounds given in Lemma \ref{lem:min_set_union} with equality, i.e., $\forall m \in [b]$,  $f_m = e_m$, where $f_m$ is as described by \eqref{eq:fm} and $e_m$ is as defined recursively by \eqref{eq:e_def1} and \eqref{eq:e_def2}.
\end{thm}
\begin{proof}
See Appendix \ref{app:supports_achievability}.
\end{proof}

\vspace{0.1in}

We use the following lemma to argue that the $m^{\text{th}}$ GHW of $\mathcal{B}_{0}$ is indeed given by $f_m$ i.e, any other $m$ dimensional subspace of $\mathcal{B}_{0}$ (i.e, other than those generated by $m$ subsets of the basis vectors) will have a support whose cardinality is no less than $f_m$. 

\vspace{0.1in}

\begin{lem}\label{lem:GHW_achievability_by_subsets}
Let $\mathcal{D}$ denote an $[n, t]$ linear code and let $\{{\bf v}_1,  \ldots, {\bf v}_t\}$ be a basis for the
code $\mathcal{D}$. Also, let $R_i = \text{supp}({\bf v}_i)$ and suppose that the sets $\{R_i\}$ are such that
\begin{enumerate}[(a)]
 \item $|R_i \cap R_j| \leq 1, \ \forall i, j, i \neq j$,
 \item any element $\ell \in [n]$ belongs to at most two sets among the sets $\{R_i\}$, and
 \item $|R_i \backslash \cup_{\substack{j=1 \\ j \neq i}}^{t}R_j| \  \geq \ 1$.
\end{enumerate}
Then the generalized Hamming weights of the code $\mathcal{D}$ are given by
\begin{eqnarray}
 d_m(\mathcal{D}) & = & \min_{\substack{\mathcal{I} \subseteq [t] \\ |\mathcal{I}|=m}} \left| \cup_{i \in
\mathcal{I}}R_i\right|, \ \ 1 \leq m \leq t.
\end{eqnarray}
\end{lem}
\begin{proof}
See Appendix \ref{app:GHW_achievability_by_subsets}.
\end{proof}

\vspace{0.1in}

We now note that Lemma \ref{lem:GHW_achievability_by_subsets} is readily applicable to the code $\mathcal{B}_{0}$ obtained via the Tur\'an graph construction. From this we conclude that the GHWs of this code $\mathcal{B}_{0}$ achieve the upper bounds given by Theorem \ref{thm:GHW_upper_bounds_C0}. 

\section{A New Upper Bound on Minimum Distance for the Single Erasure Case} \label{sec:all_symbol_single}

The approach described in Section \ref{sec:bounds} directly applies to the setting of codes with all-symbol locality which can handle single erasures. This results in a new upper bound on $d_{\min}$ for this class of codes which is in general tighter than that given by \eqref{eq:gopalan_bound}. Let $\mathcal{C}$ be an $[n, k, d_{\min}]$ code having $(r, \delta = 2)$ all-symbol locality, i.e., any code symbol is covered by a local parity. As with locally reconstructible codes, consider the subcode $\mathcal{B}_{0}$ of $\mathcal{C}^{\perp}$ which is obtained as the span of all codewords of Hamming weight less than or equal to $r+1$, i.e., $\mathcal{B}_{0}  =  \text{span}\left({\bf c} \in \mathcal{C}^{\perp}, |\text{supp}({\bf c})| \leq r+1 \right)$. It is easy to see that $\text{dim}(\mathcal{B}_{0}) \ \geq \ \frac{n}{r+1}$. Lemma \ref{lem:min_set_union} can now be applied to this $\mathcal{B}_{0}$ which enables us to upper bound the GHWs of $\mathcal{B}_{0}$ and in turn, upper bound the minimum distance of $\mathcal{C}$. 

\vspace{0.1in}

\begin{thm} \label{thm:GHW_upper_bounds_C0_all_symbol}
Let $b = \left\lceil \frac{n}{r+1}\right\rceil$. Then, the first $b$ generalized Hamming weights of the subcode
$\mathcal{B}_{0}$ defined above are upper bounded by $d_m(\mathcal{B}_{0}) \leq e_m, \ 1 \leq m \leq b$, where $e_m$ is as recursively defined by $e_b \ = \ n$, and 
\begin{eqnarray}
e_{m-1} & = & e_{m} - \left\lceil \frac{2e_m}{m}\right\rceil + (r+1),  \ 2 \leq m \leq b.
\end{eqnarray}
Furthermore, if we let $\ell$ to denote the unique integer such that
$e_{\ell} < k+\ell < e_{\ell+1}$, the minimum distance of the all-symbol locality code $\mathcal{C}$ is upper bounded by
\begin{eqnarray}
d_{\min}(\mathcal{C}) & \leq & n + 1 - (k + \ell). \label{eq:all_sym_new}
\end{eqnarray}
\end{thm}
\begin{proof} Similar to the proof of Theorem \ref{thm:GHW_upper_bounds_C0}.
\end{proof}

\vspace{0.1in}

In order to compare the upper bound given by  \eqref{eq:all_sym_new} with that given by \eqref{eq:gopalan_bound}, we note
the bound given by \eqref{eq:gopalan_bound} can be obtained by first upper bounding the GHWs of $\mathcal{B}_{0}$ by
\begin{equation} \label{eq:GHW_all_symbol_single}
 d_m(\mathcal{B}_{0}) \leq m(r+1), 1 \leq m \leq b - 1 , \ \text{and} \ d_b(\mathcal{B}_{0}) \leq  n, 
\end{equation}
where $b = \left\lceil\frac{n}{r+1}\right\rceil$, and then calculating the $k^{\text{th}}$ gap based on these bounds. But from the discussion in Section \ref{sec:bounds} (see \eqref{eq:comparision_bounds_single}),  we know that the bounds on GHWs of $\mathcal{B}_{0}$ given by Theorem \ref{thm:GHW_upper_bounds_C0_all_symbol} are, in general, tighter than the bounds in \eqref{eq:GHW_all_symbol_single} and hence we conclude that the minimum distance bound given by \eqref{eq:all_sym_new} is also tighter, in general, than that given by \eqref{eq:gopalan_bound}. We would, however, like to remark that it is always possible \cite{TamBar} to achieve a minimum distance which is at most one less than that suggested by the upper bound in \eqref{eq:gopalan_bound}. In Fig. \ref{fig:is_vs_as}, we plot the two bounds as a function of dimension $k$, for the case when $n = 18$ and $r = 3$. 
\begin{center}
\begin{figure} [h!]
\begin{center}
\includegraphics[width=3.5in]{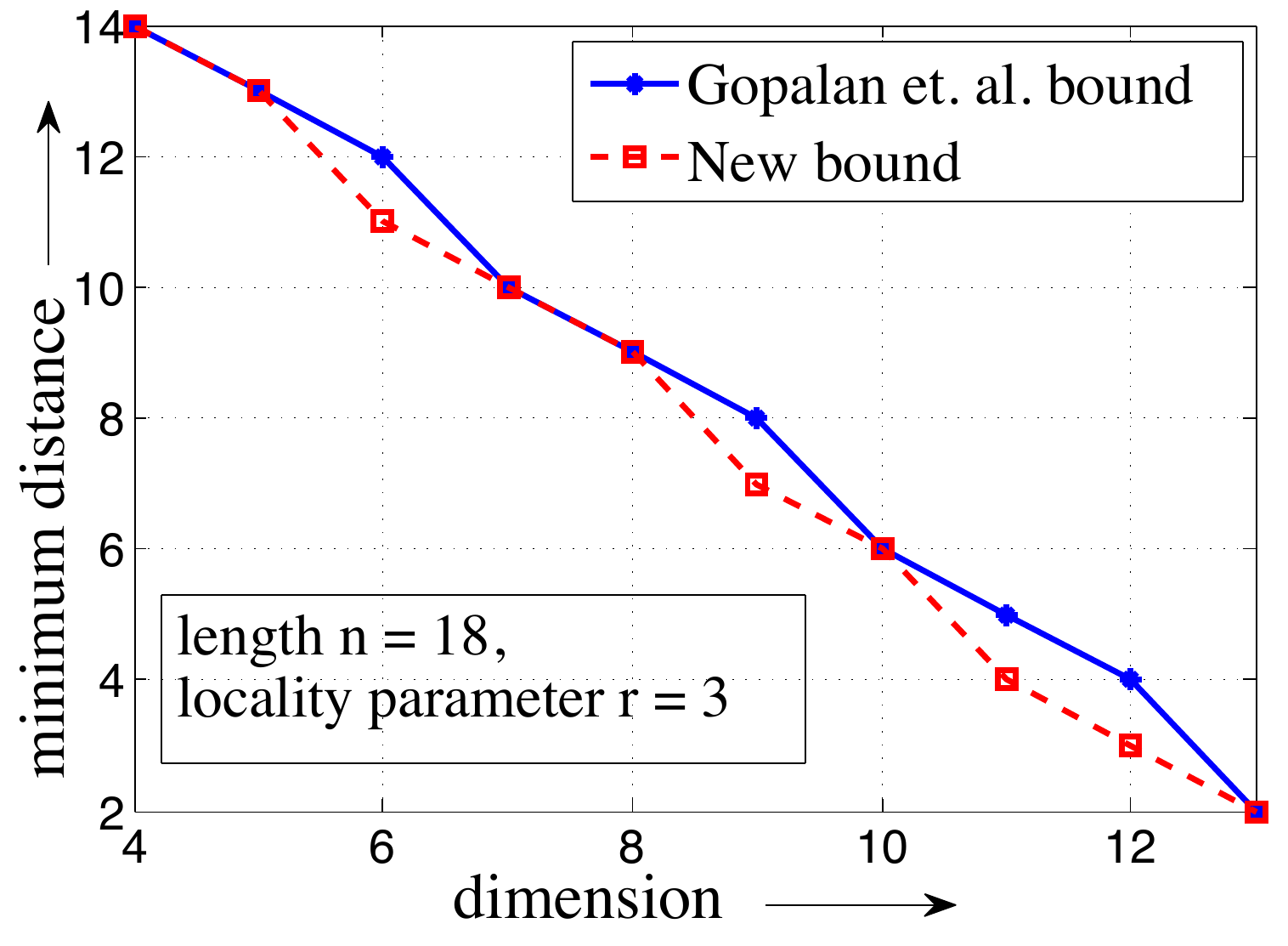}
\end{center}
\caption{Comparing the bounds on $d_{\min}$ for varying $k$ for codes with information and all-symbol locality, with $n=18$ and $r=3$.}
\label{fig:is_vs_as}
\end{figure}
\end{center}

\bibliographystyle{IEEEtran}
\bibliography{multi_coverage_locality}

\appendices

\section{Proof of Theorem \ref{thm:GHW_k-core_dmin}} \label{app:GHW_k-core_dmin}
 Let $\mathcal{B}_{0}$ denote an $[n, t]$ code and let $\mathcal{C}$ denote an $[n, k]$ code, $k \leq n - t$, such that
 \begin{enumerate}[(a)]
  \item $\mathcal{B}_{0} < \mathcal{C}^{\perp}$, and
  \item any $S$ which is a $k$-core of $\mathcal{B}_{0}$ is also a $k$-core of $\mathcal{C}^{\perp}$.
 \end{enumerate}
 We claim that for any set $S \subseteq [n]$, $|S| = g_k(\mathcal{B}_{0})$, there exists $S' \subseteq S$,
 $|S'|=k$ such that $S'$ is a $k$-core of $\mathcal{B}_{0}$.  Supposing that this is true, it then means that for any set
 $S \subseteq [n]$, $|S| = g_k(\mathcal{B}_{0})$, $\text{rank}\left(G|_S\right)  =  k$. Clearly, this would imply that
 $d_{\min}(\mathcal{C}) \geq n - |S| + 1 = n - g_k(\mathcal{B}_{0}) + 1$. However, since $d_{\min}(\mathcal{C})  = n -
 g_k(\mathcal{C}^{\perp}) + 1 \leq n - g_k(\mathcal{B}_{0}) + 1$, we conclude that  $d_{\min}(\mathcal{C})  =  n + 1 - g_{k}(\mathcal{B}_{0})$. 

We will now prove our claim  that for any set $S \subseteq [n]$, $|S| = g_k$, there exists $S' \subseteq S$, $|S'|=k$ such that $S'$ is a $k$-core of  $\mathcal{B}_{0}$. Toward this, let $\mathcal{C}'$ denote the code $\mathcal{B}_{0}$ shortened to the set $S$. We note that  $\text{dim}(\mathcal{C}') \leq g_k - k$. To see why this is true, if we suppose that $\text{dim}(\mathcal{C}') > g_k - k$,  then it will imply that $d_{g_k - k + 1}(\mathcal{B}_{0}) \leq g_k$, where $d_{g_k - k+ 1}(\mathcal{B}_{0})$ denotes the $(g_k  - k +1)^{\text{th}}$ Generalized Hamming weight of $\mathcal{B}_{0}$. But this contradicts the fact that the number of GHWs of  $\mathcal{B}_{0}$ till $g_k$ is exactly $g_k - k$ and hence we get that  $\text{dim}(\mathcal{C}') \leq g_k - k$. Now, let  $\rho$ denote the dimension of $\mathcal{C}'$ and let $H' = [ I_{\rho} |P_{\rho \times (g_k - \rho)}]$ denote  a generator  matrix of $\mathcal{C}'$, up to permutation of columns. If $T$ denotes the support of the matrix $P_{\rho \times (g_k -  \rho)}$, then any $S' \subseteq T$, $|S'| = k$ is $k$-core of $\mathcal{B}_{0}$.

\section{Proof of Theorem \ref{thm:GHW_k-core}} \label{app:GHW_k-core}
 Let $\mathcal{B}_{0}$ denote an $[n, t]$ code and let $\mathcal{C}$ denote an $[n, k]$ code, $k \leq n - t$, such that
 \begin{enumerate}[(a)]
  \item $\mathcal{B}_{0} < \mathcal{C}^{\perp}$, and
  \item any $S$ which is a $k$-core of $\mathcal{B}_{0}$ is also a $k$-core of $\mathcal{C}^{\perp}$.
 \end{enumerate}
Note that Theorem \ref{thm:GHW_k-core_dmin} implies that $g_k(\mathcal{C}^{\perp}) =  g_k(\mathcal{B}_{0})$. This determines the last $n - g_k(\mathcal{B}_{0})$ GHWs of $\mathcal{C}^{\perp}$ and are given by
 \begin{eqnarray}
 d_i(\mathcal{C}^{\perp}) & = & i + k, \ g_k(\mathcal{B}_{0}) - k +1 \leq i \leq n - k.
 \end{eqnarray}

Assuming that $g_k(\mathcal{B}_{0}) - k \geq 1$, it now remains to be proved that 
\begin{eqnarray} \label{eq:temp_proof_GHW_kcores}
d_i(\mathcal{C}^{\perp}) & = & d_i(\mathcal{B}_{0}), \ 1 \leq i \leq g_k(\mathcal{B}_{0}) - k,
\end{eqnarray}
i.e., the first $g_k(\mathcal{B}_{0}) - k$ GHWs of $\mathcal{C}^{\perp}$ are exactly same as those of $\mathcal{B}_{0}$.  Toward this, we first note that any set $S$ which is a $b$-core of $\mathcal{B}_{0}$ is also a $b$-core of $\mathcal{C}^{\perp}$, for any  $b$ such that $b < k$. This is because for any $S$ which is a $b$-core of $\mathcal{B}_{0}$ ($b < k$), there exists a $k$-core $S'$ of $\mathcal{B}_{0}$ such that $S \subseteq S'$. We also claim that for any set $S \subseteq [n]$, $|S| = g_b$, there exists $S' \subseteq S$, $|S'|=b$ such that $S'$ is a $b$-core of  $\mathcal{B}_{0}$. Proof is similar to the claim regarding $k$-cores which appeared in the proof of Theorem \ref{thm:GHW_k-core_dmin}. 

We will now prove \eqref{eq:temp_proof_GHW_kcores} via induction starting at $i=1$. Let us denote $\ell = g_k(\mathcal{B}_{0}) - k$. 
Note that 
\begin{eqnarray} \label{eq:k_uppperbound}
k & > & (d_1(\mathcal{B}_{0}) - 1 ) + (d_2(\mathcal{B}_{0}) - d_1(\mathcal{B}_{0})- 1 ) + \ldots + (d_{\ell}(\mathcal{B}_{0}) - d_{\ell-1}(\mathcal{B}_{0})- 1 ).
\end{eqnarray}
Now, suppose that $d_1(\mathcal{C}^{\perp}) \leq d_1(\mathcal{B}_{0})-1$. Clearly any $S$ such that $|S|=d_1(\mathcal{B}_{0})-1$ is an $|S|$-core of $\mathcal{B}_{0}$. From \eqref{eq:k_uppperbound}, we see that $d_1(\mathcal{B}_{0}) - 1 < k$ and hence $S$ is also an $|S|$-core of $\mathcal{C}^{\perp}$. But this contradicts the assumption that $d_1(\mathcal{C}^{\perp}) \leq d_1(\mathcal{B}_{0})-1$ and hence we conclude that $d_1(\mathcal{C}^{\perp}) = d_1(\mathcal{B}_{0})$. Next, assume that $d_{i}(\mathcal{C}^{\perp}) = d_{i}(\mathcal{B}_{0})$, for some $i$ such that $1 \leq i \leq \ell-1$. We will now prove that $d_{i+1}(\mathcal{C}^{\perp}) = d_{i+1}(\mathcal{B}_{0})$. We consider the following cases:
\begin{enumerate}[(a)]
\item $d_{i+1}(\mathcal{B}_{0}) = d_{i}(\mathcal{B}_{0}) + 1$. In this case, note that 
\begin{eqnarray}
d_i(\mathcal{B}_{0}) & \stackrel{(i)}{=} & d_i(\mathcal{C}^{\perp}) \\
& \stackrel{(ii)}{<} & d_{i+1}(\mathcal{C}^{\perp}) \\
& \stackrel{(iii)}{\leq} & d_{i+1}(\mathcal{B}_{0}) \\
& = & d_{i}(\mathcal{B}_{0}) + 1, 
\end{eqnarray}
where $(i)$ follows from the induction hypothesis and $(ii)$ follows from \eqref{eq:ordering_GHW}.
This then implies that $(iii)$ must be an equality, i.e., $d_{i+1}(\mathcal{C}^{\perp}) = d_{i+1}(\mathcal{B}_{0})$.
\item $d_{i+1}(\mathcal{B}_{0}) > d_{i}(\mathcal{B}_{0}) + 1$. In this case, if we let $m = d_{i+1}(\mathcal{B}_{0}) - (i+1)$, note that 
\begin{eqnarray}
g_{m}(\mathcal{B}_{0}) & = & d_{i+1}(\mathcal{B}_{0}) - 1.
\end{eqnarray}
Now, if $S$ is any set such that $|S| = g_{m}(\mathcal{B}_{0})$, then there exists $S' \subseteq S$, $|S'| = m$ and $S'$ is an $m$-core of $\mathcal{B}_{0}$. From \eqref{eq:k_uppperbound}, we see that $m < k$ and hence $S'$ is also an $m$-core of $\mathcal{C}^{\perp}$. Now, without loss of generality, suppose that  $d_{i+1}(\mathcal{C}^{\perp}) = d_{i+1}(\mathcal{B}_{0})-1$. Also, let $\mathcal{D}$ denote an $(i+1)$-dimensional subcode of $\mathcal{C}^{\perp}$ having support $S_{\mathcal{D}}$ such that $|S_{\mathcal{D}}| = d_{i+1}(\mathcal{C}^{\perp})$. Note that 
for any set $T \subseteq S_{\mathcal{D}}$ such that $|T| = |S_{\mathcal{D}}| - i =  m$, one can find a non-zero vector in $\mathcal{D}$ whose support is fully contained within $T$ and hence there cannot exist an $m$-core of $\mathcal{C}^{\perp}$ within $S_{\mathcal{D}}$. However, we know that this is not true and hence we conclude that $d_{i+1}(\mathcal{C}^{\perp}) = d_{i+1}(\mathcal{B}_{0})$.
\end{enumerate}

\section{Proof of Theorem \ref{thm:dim_local_code}} \label{app:dim_local_code}
Consider the code
\begin{eqnarray} 
\mathcal{B}_{0} & = & \text{span}\left({\bf c} \in \mathcal{C}^{\perp}, |\text{supp}({\bf c})| \leq r+1 \right)
\end{eqnarray}
and let $\mathcal{B} = \{\bf{c}_1, \ldots, {\bf c}_b\}$ denote a basis for $\mathcal{B}_{0}$ such that
$|\text{supp}({\bf c}_i)| \leq r+1, \forall i \in [b]$. Also, let $S_i = \text{supp}({\bf c}_i)$. Define the quantity
\begin{eqnarray}
s_i & = & \left|S_i \backslash \cup_{\substack{j = 1\\j \neq i}}^{b}S_j\right|, \ 1 \leq i \leq b.
\end{eqnarray}
We claim that for the code $\mathcal{C}$ to be locally reconstructible, it must necessarily be true that $s_i \leq 1,
\forall i \in [b]$. To see this, suppose that for some $i$, $s_i \geq 2$ and let $\{\ell_1, \ell_2\} \subseteq S_i \backslash
\cup_{\substack{j = 1\\j \neq i}}^{b}S_j$. Then, if $A_{\ell_1}$ and $A_{\ell_2}$, respectively denote all the
local parities covering the code symbols $c_{\ell_1}$ and $c_{\ell_2}$, it would mean that $A_{\ell_1} = A_{\ell_2}$. This is
because  $\mathcal{B}$ is a basis and any linear combination whose support contains ${\ell_1}$ will also contain ${\ell_2}$.
The claim now follows by noting from Lemma \ref{lem:loc_char_A} that the code cannot be locally reconstructible unless
$A_{\ell_1} \neq A_{\ell_2}$.

In order to proceed and complete the proof of the theorem, we note that $n-\sum_{i=1}^{b}s_i$ code symbols are covered by
more than one of the sets $S_i, i \in [b]$ and hence it must be true that
\begin{eqnarray}
 \sum_{i=1}^{b}s_i + 2\left(n-\sum_{i=1}^{b}s_i\right) & \leq & b(r+1) \\
\implies 2n - \sum_{i=1}^{b}s_i & \leq & b(r+1) \\
\implies 2n - b & \leq & b(r+1) \label{eq:temp1_proof_dim_local_code}\\
\implies b \geq \frac{2n}{r+2},
\end{eqnarray}
where \eqref{eq:temp1_proof_dim_local_code} follows since $s_i \leq 1, \forall i \in [b]$.

\section{Proof of Lemma \ref{lem:min_set_union}} \label{app:min_set_union}
We will prove the lemma via induction, starting at $m =b$ and decrementing $m$ at each step. Clearly $f_b \leq e_b = n$.
Now assuming that for some $m, 2 \leq m \leq b$, $f_m \leq e_m$, we will prove that $f_{m-1} \leq e_{m-1}$. Without loss of
generality, let $\{S_i,  1 \leq i \leq m\}$ be such that $|\cup_{i=1}^{m}S_i| = f_m$. Define
\begin{eqnarray}
s_i & = & \left|S_i \backslash \cup_{\substack{j = 1\\j \neq i}}^{m}S_j\right|, \ 1 \leq i \leq m.
\end{eqnarray}
Then, noting that $n-\sum_{i=1}^{m}s_i$ elements are covered by more than one of the sets $S_i, i \in [m]$, we get that
\begin{eqnarray}
\sum_{i=1}^{m}s_i + 2\left(f_m-\sum_{i=1}^{m}s_i\right) & \leq & m(r+1) \\
\implies \sum_{i=1}^{m}s_i & \geq & 2f_m - m (r+1).
\end{eqnarray}

Also, let $s* = \max_{i \in [m]}s_i$ and without loss of generality, assume that $s_1 = s*$. Note that in this case, $s_1
\geq \frac{\sum_{i=1}^{m}s_i}{m}$. Now, if we consider union of the sets $\{S_i, \ 2 \leq i \leq m\}$, we get that
\begin{eqnarray}
 |\cup_{i=2}^{m}S_i| & = & f_m - s_1 \\
& \leq & f_m - \frac{\sum_{i=1}^{m}s_i}{m} \\
& \leq & f_m -  \frac{2f_m-m(r+1)}{m} \\
& = & \frac{m-2}{m}f_m + (r+1) \\
& \leq & \frac{m-2}{m}e_m + (r+1) \\
& = & e_m - \frac{2e_m}{m} +(r+1), \\
\end{eqnarray}
which implies that $f_{m-1}  \leq  e_m - \frac{2e_m}{m} +(r+1)$. Finally, noting that $f_{m-1}$ is an integer, we get
that
\begin{eqnarray}
 f_{m-1} & \leq & e_m - \left\lceil \frac{2e_m}{m} \right\rceil +(r+1) \ = \ e_{m-1}
\end{eqnarray}

\section{Proof of Theorem \ref{thm:supports_achievability}} \label{app:supports_achievability}
Consider any set of $m$ parities, say $\{\underline{p}_1, \ldots, \underline{p}_m \}$ having supports $S_1, \ldots S_m$. Note that by our construction the parity $\underline{p}_i$ corresponds to the vertex $i$ in the Tur\'an graph. The cardinality of union of the supports $S_1, \ldots S_m$ can be calculated from the graph as $|\cup_{i=1}^{m}S_i|  =  m + |E|$, where $E$ is the set of all the edges in the graph with at least one of the end points being a vertex belonging to the set $[m]$. The quantity $|E|$ can be equivalently be  computed by first counting the number of edges in the graph restricted to the remaining  vertices $\{m+1, m+2, \ldots, r+\beta\}$ and then subtracting it from the total number of edges in the original graph. Thus, for calculating $f_m$, it is sufficient to find a restricted graph on $r + \beta -m$ vertices having the maximum number of edges. Let $r+\beta-m=ux+v, 0 \leq u \leq \beta -1, 0 \leq v \leq x-1$. Then, it is easy to see that  the number of edges in a restricted graph on $r + \beta -m$ vertices is maximized if the restricted graph consists of  $u+1$ vertices each from any $v$ out of the $x$ partitions and $u$ vertices each from the remaining $x-v$ partitions. 

It is straightforward to see now that the difference $f_m - f_{m-1}, 2 \leq m \leq r+\beta$, is given by 
\begin{eqnarray} \label{eq:fm-fm-1}
f_m - f_{m-1} & = & (u+1)v + u(x-v-1) + 1 \nonumber \\
& = & v + ux-u +1. \label{eq:fm-fm-1}
\end{eqnarray}
The expression in \eqref{eq:fm-fm-1} is evaluated for $m, 2 \leq m \leq r+\beta$ and is shown in Fig. \ref{fig:pole_structure}. For the array given in Fig. \ref{fig:pole_structure}, we  number the rows from $0$ to $\beta -1 $ and the columns from $0$ to $x-1$. Then, the value of $f_m - f_{m-1}$ is simply the $(u, v)^{\text{th}}$ entry in this array.
\begin{center}
\hspace*{2.0in}
\begin{figure} [h!]
\begin{center}
\includegraphics[width=4.5in]{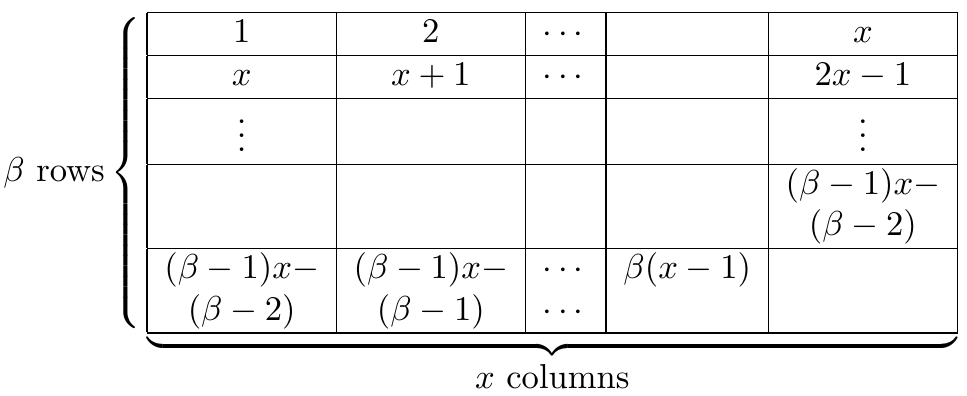}
\end{center}
\caption{$f_m-f_{m-1}$, $r+\beta \geq m \geq 2$ of $\mathcal{B}_{0}$ obtained via Tur\'an graph construction, where the sequence of differences are in the descending order and the matrix has to be read row after row from left to right.}
\label{fig:pole_structure}
\end{figure}
\end{center}
Now, we will show that the sequence $\{f_m\}$ as defined by \eqref{eq:fm-fm-1}, satisfies the recursion given in \eqref{eq:e_def1} and \eqref{eq:e_def2}, i.e., $f_{r+\beta}=n$ and
\begin{equation}
 f_m - f_{m-1} = \left \lceil \frac{2f_m}{m} \right \rceil  - (r+1).
\end{equation}
Since the sequence $\{e_m\}, 1 \leq m \leq b=r+\beta$ is unique  (given that $e_m = n$), it then follows that $f_m = e_m, 1 \leq m \leq b=r+\beta$, which will complete our proof. We begin by noting that 
\begin{eqnarray} \label{eq:fm_temp}
f_m & = & n - \sum_{i=r+\beta}^{m+1} (f_i - f_{i-1}).
\end{eqnarray}
Next, note that the sum $\sum_{i=r+\beta}^{m+1} (f_i - f_{i-1})$ can be calculated from the array in Fig. \ref{fig:pole_structure} as sum of the first $r+\beta -m$ entries, where the entries are read from left to right in each row and the rows are read from top to bottom, i.e.,
\begin{eqnarray} \label{eq:usearray}
\sum_{i=r+\beta}^{m+1} (f_i - f_{i-1}) & = & \sum_{i=1}^{ux-u+v}i + \sum_{i=1}^{u}(ix-i+1),
\end{eqnarray}
where the first term on the R.H.S counts all the unique elements once and the second term on the R.H.S counts the repeated terms. Combining \eqref{eq:fm_temp} and \eqref{eq:usearray}, we get that
\begin{eqnarray}
 \frac{2f_m}{m} & = & \frac{(r+\beta)(r+2) - (ux+v-u)(ux+v-u+1) - u(u+1)x + u(u-1)}{(r+\beta)-(ux+v)} \\
              & = & \frac{[(r+\beta)-(ux+v)][(r+1)+(ux+v-u)] + (r+\beta) -ux - v\beta + uv}{(r+\beta) - (ux+v)} \\
              & = & \frac{[(r+\beta)-(ux+v)]q' + r'}{(r+\beta) - (ux+v)},
\end{eqnarray}
where $q' = (r+1)+(ux+v-u)$ and $r' = (r+\beta) -ux - v\beta + uv$. It is straightforward to check that $1 \leq r' \leq (r+\beta)-(ux+v)$, which implies that 
\begin{equation} \label{eq:vneq0}
\left \lceil \frac{2f_m}{m} \right \rceil = q'+1 =  r+1 + ux +v -u+1.
\end{equation}
Combining \eqref{eq:fm-fm-1} with \eqref{eq:vneq0}, we finally get that 
\begin{equation}
 f_m - f_{m-1} = \left \lceil \frac{2f_m}{m} \right \rceil  - (r+1).
\end{equation}

\section{Proof of Lemma \ref{lem:GHW_achievability_by_subsets}} \label{app:GHW_achievability_by_subsets}

Consider an $m$ dimensional subcode $\mathcal{D}'$ of $\mathcal{D}$ and let $\mathcal{D}'$ have a basis
$\{ {\bf u_1}, \ldots,  {\bf u_m}\}$. Without loss of generality, let us assume that the basis of $\mathcal{D}'$ is obtained
as
\begin{eqnarray} \label{eq:matrix_B}
 \left[ \begin{array}{c}{\bf u_1} \\ \vdots \\ {\bf u_m} \end{array} \right] & = & \left[ \begin{array}{c|c} I_m & B_{m
\times (b-m)} \end{array} \right]  \left[ \begin{array}{c}{\bf v_1} \\ \vdots \\ {\bf v_m} \\ {\bf v_{m+1}}\\ \vdots \\ {\bf
v_{b}} \end{array}\right].
\end{eqnarray}
Also, let $\{R_i', 1 \leq i \leq m\}$ denote the supports of the vectors $\{ {\bf u_1}, \ldots,  {\bf u_m}\}$.
We claim that $|\cup_{i=1}^{m}R_i'| \geq |\cup_{i=1}^{m}R_i|$. To see this, we consider
any element $x \in \cup_{i=1}^{m}R_i$ and examine what happens to it when the $m$ linear combinations are taken. We divide
the discussion into the following cases:
\begin{enumerate}[(a)]
\item $x \in R_i, i \leq m$ and does not belong to any other support set. Clearly, $x \in R_i'$.
\item $x \in R_i, R_j$ such that $1 \leq i < j \leq m$. By assumption, $x$ then does not belong to any other support set
and clearly in this case, $x \in R_i', R_j'$.
\item $x \in R_i, R_j$ such that $ i \leq m, j \geq m+1$. Note that $x$ then does not belong to any other support set.
Now, consider the $j^{\text{th}}$ column of the matrix $[I | B]$ and let us call it as ${\bf b}$. We consider three sub-cases for this
situation based on the column weight of ${\bf b}$.
\begin{enumerate}[(i)]
\item Column weight of ${\bf b}$ is $0$. Clearly, then $x \in R_i'$.
\item Column weight of ${\bf b}$ is $1$, say $b_{\ell} \neq  0$. Suppose, $\ell \neq i$, then $x \in R_i', R_{\ell}'$. Now if $\ell = i$, the element $x$ need not be present in $R_i'$. However, for the purposes of counting $|\cup_{i=1}^{m}R_i'|$, we could replace $x$ with $y$ where $y$ is one of the elements covered only by $R_j$ (note that the such an element exists by assumption). This works because, if this particular case does occur, we will never again have to seek one of the elements
covered only by $R_j$. This is because in order for this to happen again it must be true that there exists another element
$x' \in R_i \cap R_j$, but this is contrary to our assumption that any two support sets have intersection at most $1$.
\item Column weight of ${\bf b}$ is $2$ or more, say $b_{\ell_1}, b_{\ell_2} \neq  0$. Without loss of generality if assume
that $\ell_1 \neq i$, then $x \in S_{\ell_1}'$.
\end{enumerate}
\end{enumerate}
Thus we see that in all the cases we either do not lose the element $x$ or there is another unique element $y$ which can
compensate for $x$ while counting the support cardinality after the linear combinations are taken. Hence we conclude that
$|\cup_{i=1}^{m}R_i'| \geq |\cup_{i=1}^{m}R_i|$.

\end{document}